\documentclass[letterpaper, 10 pt, conference]{ieeeconf}

\usepackage[T1]{fontenc}
\usepackage{cite}
\usepackage{amsmath,amssymb}
\usepackage{graphicx}
\usepackage{booktabs}
\usepackage{url}
\makeatletter\let\labelindent\relax\makeatother
\usepackage{enumitem}
\usepackage{xcolor}
\usepackage{soul}
\newcommand{\ap}[1]{#1}
\newcommand{\apx}[1]{#1}
\newcommand{\vf}[1]{#1}

\newcommand{\sbl}[1]{#1}
\newcommand{\sblx}[1]{#1}
\usepackage{tikz}
\usepackage{algorithm}
\usepackage{algpseudocode}
\usepackage[hidelinks]{hyperref}
\hypersetup{
pdftitle={Underwater Barrier Sensing Under a Network False-Alarm Budget},
pdfauthor={Hesam Mosalli, Stephane Blouin and Amir G. Aghdam},
pdfkeywords={Sensor placement; minimax optimization; resource allocation;
underwater acoustics; barrier coverage}}

\IEEEoverridecommandlockouts   
\newtheorem{theorem}{Theorem}
\newtheorem{proposition}{Proposition}

\newtheorem{corollary}{Corollary}
\newtheorem{remark}{Remark}

\newcommand{\Rr}{\mathbb{R}}
\newcommand{\D}{\mathcal{D}}
\newcommand{\Cc}{\mathcal{C}}
\newcommand{\Gpath}{\mathcal{G}}
\newcommand{\Qset}{\mathcal{Q}}
\newcommand{\dd}{\,\mathrm{d}}
\newcommand{\Bs}{\Psi}
\newcommand{\Ex}{\mathbb{E}}
\newcommand{\haz}{\eta}

\newenvironment{IEEEkeywords}
{\par\vspace{0.3em}\noindent\textbf{Index Terms---}}
{\par\vspace{0.3em}}

\newcommand{\ALLOCGAIN}{14.9}

\newcommand{\CVARGAIN}{18.2}
\newcommand{\CHANW}{120}
\newcommand{\CHANL}{4}
\newcommand{\CHANDEP}{2000}
\newcommand{\FREQ}{300}
\newcommand{\ISOERR}{0.06}
\newcommand{\LANEX}{30}
\newcommand{\LANEDB}{9}
\newcommand{\LANEW}{20}
\newcommand{\NSLOTS}{20}
\newcommand{\NDEPTHS}{5}
\newcommand{\NCAND}{100}
\newcommand{\NBUD}{6}
\newcommand{\NSUBSETS}{1{,}192{,}052{,}400}
\newcommand{\ALARMRATE}{6}
\newcommand{\BUDGET}{8.333\times 10^{-3}}
\newcommand{\QUNIF}{1.389\times 10^{-3}}
\newcommand{\DTUNIF}{-7.23}
\newcommand{\EPS}{0.2}
\newcommand{\TAU}{1.609}
\newcommand{\QMAX}{1\times 10^{-2}}
\newcommand{\QMAXFA}{7.2}
\newcommand{\ACUT}{-32.3}
\newcommand{\TRUNCBOUND}{1.9\times 10^{-4}}
\newcommand{\HZ}{25}
\newcommand{\GRIDDESC}{41\times201\times81=667{,}521}
\newcommand{\SIGMA}{6}
\newcommand{\TW}{250}
\newcommand{\DWELL}{5}
\newcommand{\SPEED}{1.5}
\newcommand{\DI}{6}

\newcommand{\PSIUNIF}{3.040}
\newcommand{\PSIALLOC}{3.494}
\newcommand{\PSISPREAD}{0.874}

\newcommand{\MISSUNIF}{4.782}
\newcommand{\MISSALLOC}{3.038}
\newcommand{\CVARUNIF}{0.537}
\newcommand{\CVARALLOC}{0.635}

\newcommand{\PLACEFACTOR}{3.5}

\newcommand{\UNIFITERS}{30}
\newcommand{\JOINTLO}{3.494}
\newcommand{\JOINTHI}{23.909}

\newcommand{\RESHZ}{100}
\newcommand{\RESHZFINE}{12.5}
\newcommand{\RESSPAN}{6.2}
\newcommand{\RESLAST}{1.4}

\title{\textbf{Underwater Barrier Sensing Under a Network False-Alarm Budget}}

\author{Hesam~Mosalli, Stephane Blouin and Amir G. Aghdam
\thanks{This work was supported by Defence Research and Development Canada.}
\thanks{H. Mosalli and A. G. Aghdam are with Senseau Technologies Inc., Montr\'eal, QC, Canada (e-mail: mosalli.h@gmail.com; amir\_aghdam@yahoo.com). S. Blouin is with the Department of Electrical Engineering, \'Ecole de technologie sup\'erieure, Montr\'eal, QC, Canada (e-mail: blouin.s@etsmtl.ca).}
\thanks{The authors acknowledge the use of Grammarly and AI agents in refining the writing of this manuscript.}
}

\begin{document}
\maketitle

\begin{abstract}
We study joint sensor placement and detector operating-point allocation for guarding a region against an adversarial crossing. Each sensor may adjust its detection threshold, while all sensors draw on a shared network-level false-alarm budget. For a Gaussian signal-excess model of a square-law detector, we give a necessary and sufficient condition for the negative-log miss hazard to be concave in a sensor's false-alarm probability. At fixed deployment, this yields a concave max-min allocation whose saddle point equalizes marginal barrier gain per unit alarm against a worst-case occupation measure. A perspective reformulation supplies linear cuts for the mixed-integer joint problem, and a separate certificate lower-bounds the continuum barrier strength independently of the graph used for path search. On a deep-water case study based on the canonical Munk profile, reallocating a fixed alarm budget raises the graph design objective by \ALLOCGAIN{}\,\% at unchanged sensor count and alarm load, and the certificate establishes the barrier requirement on the interpolated continuum field, which the graph value alone cannot.
\end{abstract}


\section{Introduction}
\label{sec:intro}

Monitoring a strait or harbor often relies on a fixed network of sensors, deployed once and left to detect subsequent crossings. Full coverage of the region would require \ap{a large number of nodes to ensure} that no crossing goes undetected. The network's capability depends on two main decisions: sensor placement and the sensitivity of each detector. The former is well studied, but the latter is usually predetermined by the sensing model. This paper seeks to reintroduce detector sensitivity as a design variable, confronting an adversary who selects the crossing route.

Classical strong and $k$-barrier coverage ask whether every crossing meets one or $k$ sensing regions~\cite{Kumar2005,Wu2016}, and a large literature softens that binary predicate through graded barrier quality~\cite{Chen2009Quality} and probabilistic $\varepsilon$-barriers~\cite{Li2012,chen2013}. Underwater work carries these notions into three dimensions \ap{and} makes the sensing radius attenuation-dependent~\cite{Barr2008,Chang2019}. In parallel, path-based exposure measures, \ap{defined in terms of sensing intensity or target-detection probability,} support sensor placement that maximizes the \ap{minimum exposure along any} crossing~\cite{Meguerdichian2001Exposure,Amaldi2012}\ap{.} Hamilton-Jacobi solvers \ap{provide the corresponding} continuous best response~\cite{Gilles2020SEG}\ap{, while} seabed placement \ap{approaches for} stochastic traffic optimize expected rather than worst-case detection~\cite{Kim2025Barrier}, \ap{with recent work also incorporating} false alarms into a sensor-availability model~\cite{Kim2026ACC}.

Most of these formulations fix the detector operating point once the sensing model is specified. \ap{However, the exceptions are particularly relevant to the present work.} Threshold-based detection has been combined with optimal \ap{sensor placement} for passive acoustic diver detection~\cite{Stolkin2009}\ap{;} a scenario-dependent \ap{detection threshold has been jointly optimized} with hydrophone placement against a worst-case track~\cite{Molyboha2012}; \ap{sensor placement subject to} explicit detection and false-alarm requirements has been \ap{considered} for fusion-based surveillance~\cite{Chang2011Fusion}; and strong barrier coverage has been constructed under \ap{minimum detection-probability} maximum system \ap{false-alarm-probability constraints, with} a common \ap{detection threshold adjusted according} to the number of active nodes~\cite{Zhang2016ICC}.

The coupling addressed in this work is distinct. Each receiver computes a detection statistic \ap{over each} processing interval and declares \ap{a detection when the statistic exceeds a specified threshold. Lowering this} threshold improves \ap{the detection} probability across all ranges \ap{but} also increases the false-alarm rate. \ap{Importantly, this tradeoff arises at the network level:} alarms are \ap{drawn} from a \ap{single shared quota, as a} watch officer, fusion center, or downstream tracker can tolerate only a finite number of nuisance alarms per hour. \ap{Consequently, the network-wide} alarm budget imposes a linear constraint on the aggregate false-alarm probabilities of all deployed sensors. Detection thresholds are therefore allocated \ap{jointly} with sensor placement against the continuous worst-case \ap{crossing. In this formulation,} the false-alarm \ap{budget} becomes a \ap{divisible resource allocated among individual sensors,} rather \ap{than being enforced through} a common threshold calibrated \ap{to satisfy} a global constraint.

\ap{From an optimization perspective, this formulation can be viewed as} an interdiction problem: a divisible \ap{resource budget is allocated to increase the cost of an adversary's least-cost} path. Shortest-path interdiction is well \ap{established}~\cite{Israeli2002,Tayyebi2023}, but \ap{typically considers} discrete graphs with abstract interdiction costs. The closest control-theoretic analogue allocates detector dwell time across sites \ap{subject to per-sensor budgets}~\cite{LeNy2009}. \ap{In contrast, in the present formulation,} false-alarm probability is \ap{allocated from a single shared budget, the relationship between budget allocation and path cost is determined} by detection theory rather than assumed, and the \ap{adversary's path} problem is continuous.

\ap{This work (i)} formulates the joint problem of sensor placement and operating-point allocation under a network-level false-alarm constraint; \ap{(ii) establishes} a necessary and sufficient curvature condition for concavity of the negative-log miss hazard; \ap{(iii) characterizes} the allocation \ap{for a fixed sensor placement through a saddle-point formulation} and a marginal-allocation principle; \ap{(iv) develops} a perspective reformulation to bracket the joint problem; and \ap{(v) provides} a certificate that validates the continuum barrier field independently of the discretization. A deep-water case study demonstrates the quantitative improvement achievable by reallocating a \ap{fixed-alarm} budget.

\vf{The remainder of the paper is organized as follows. Section~\mbox{\ref{sec:problem}} formulates the deployment, the detection model and the design problem. Section~\mbox{\ref{sec:structure}} establishes the curvature condition, the concavity-preserving truncation, and the saddle-point and water-filling characterization of the allocation. Section~\mbox{\ref{sec:algorithm}} develops the perspective reformulation and the continuum certificate. Section~\mbox{\ref{sec:case}} applies the method to a deep-water chokepoint, reporting the truncation parameters, the design comparison and the continuum certificates. Section~\mbox{\ref{sec:conclusion}} concludes and states the limitations of the present treatment.}

\section{Problem Statement}
\label{sec:problem}

\subsection{Deployment and crossings}

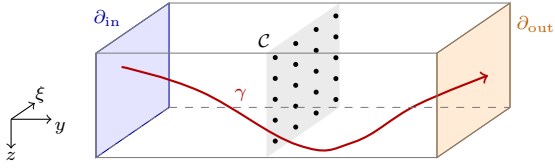
\begin{figure}[t]
\centering
\begin{tikzpicture}[x=1cm,y=1cm,scale=0.92,
                    every node/.style={font=\scriptsize}]
  \def\Wy{4.9}\def\Hz{1.5}\def\ox{1.05}\def\oy{0.72}
  \fill[blue!30,fill opacity=0.35]
        (0,0) -- (\ox,\oy) -- (\ox,-\Hz+\oy) -- (0,-\Hz) -- cycle;
  \fill[orange!45,fill opacity=0.35]
        (\Wy,0) -- (\Wy+\ox,\oy) -- (\Wy+\ox,-\Hz+\oy) -- (\Wy,-\Hz) -- cycle;
 
  \draw[blue!60!black] (0,0) -- (\ox,\oy) -- (\ox,-\Hz+\oy) -- (0,-\Hz) -- cycle;
  \draw[orange!75!black]
        (\Wy,0) -- (\Wy+\ox,\oy) -- (\Wy+\ox,-\Hz+\oy) -- (\Wy,-\Hz) -- cycle;
  \node[blue!60!black,anchor=east] at (0.5,.5) {$\partial_{\mathrm{in}}$};
  \node[orange!75!black,anchor=west] at (\Wy+\ox-0.04,-0.30+\oy)
        {$\partial_{\mathrm{out}}$};
  \fill[black!8] (0.5*\Wy,0) -- (0.5*\Wy+\ox,\oy)
                 -- (0.5*\Wy+\ox,-\Hz+\oy) -- (0.5*\Wy,-\Hz) -- cycle;
  \foreach \t in {0.12,0.40,0.68,0.95}
    \foreach \d in {0.10,0.30,0.50,0.70,0.90}
       \fill[black] (0.5*\Wy+\t*\ox,-\d*\Hz+\t*\oy) circle (0.036);
  \node[anchor=south east,inner sep=1pt]
        at (0.5*\Wy+0.10*\ox,-0.10*\Hz+0.12*\oy+0.1) {$\Cc$};
  \draw[red!70!black,thick,->] plot[smooth,tension=0.7] coordinates
       {(0.368,-0.198) (1.50,-0.55) (2.90,-1.32) (3.70,-1.28)
        (4.60,-0.75) (5.635,-0.321)};
  \node[red!70!black,anchor=south] at (2.08,-0.82) {$\gamma$};
  \draw[gray,dashed] (\ox,-\Hz+\oy) -- (\Wy+\ox,-\Hz+\oy);
  \draw[gray,dashed] (0,-\Hz) -- (\ox,-\Hz+\oy);
  \draw[gray] (0,0) rectangle (\Wy,-\Hz);
  \draw[gray] (\ox,\oy) -- (\Wy+\ox,\oy) -- (\Wy+\ox,-\Hz+\oy);
  \draw[gray] (0,0) -- (\ox,\oy);
  \draw[gray] (\Wy,0) -- (\Wy+\ox,\oy);
  \draw[gray] (\Wy,-\Hz) -- (\Wy+\ox,-\Hz+\oy);
  \begin{scope}[shift={(-1.22,-0.95)}]
    \draw[->] (0,0) -- (0.58,0) node[anchor=north west,inner sep=1pt] {$y$};
    \draw[->] (0,0) -- (0.34,0.233) node[anchor=south west,inner sep=0pt] {$\xi$};
    \draw[->] (0,0) -- (0,-0.42) node[anchor=north,inner sep=1pt] {$z$};
  \end{scope}
\end{tikzpicture}
\caption{The guarded domain $\D$, schematic: a crossing $\gamma$ from
$\partial_{\mathrm{in}}$ to $\partial_{\mathrm{out}}$, and candidate sites $\Cc$.}
\label{fig:domain}
\end{figure}

Let $\D\subset\Rr^{3}$ be the guarded body of water, a bounded \sbl{domain with compact} closure $\overline{\D}$ and coordinates $u=(y,\xi,z)$, where $y$ is the crossing direction, $\xi$ is the cross-channel coordinate and $z$ is depth. The intruder \ap{enters the domain} through $\partial_{\mathrm{in}}=\partial\D\cap\{y=0\}$ and \ap{aims to} reach $\partial_{\mathrm{out}}=\partial\D\cap\{y=L\}$, \ap{with both sets assumed to be} nonempty. Admissible crossings are the Lipschitz paths $\gamma:[0,1]\to\overline{\D}$ \ap{connecting the two boundary faces;} that is $\Gpath=\{\gamma\in\mathrm{Lip}([0,1];\overline{\D}):\gamma(0)\in\partial_{\mathrm{in}},\ \gamma(1)\in\partial_{\mathrm{out}},\ \ell(\gamma)\le\ell_{\max}\}$, \ap{where} $\ell(\gamma)$ \ap{denotes the arclength of \mbox{$\gamma$}, and} $\ell_{\max}$ \ap{is a} transit-endurance bound \ap{introduced only to ensure} compactness in Theorem~\ref{thm:saddle}. Sensors may be moored only \ap{at the} $M$ \ap{candidate sites} $\Cc=\{c_1,\dots,c_M\}$, and a deployment \ap{is represented by} $x\in\{0,1\}^{M}$\ap{, subject to} $\sum_m x_m\le n$. Fig.~\ref{fig:domain} \ap{illustrates} the setting. The intruder \ap{travels} at a fixed speed $v$, \ap{assumed common across} all sensors. Finally, $\omega$ denotes the ocean state, \ap{comprising} a sound-speed profile \ap{and} a noise field.

\subsection{Detection threshold as a decision variable}

Consider an incoherent square-law detector \ap{over} a dwell \sbl{time} of \ap{duration} $\Delta t$ and bandwidth $W$, so \ap{that the detection} statistic sums $N=2\Delta t\,W$ real samples. Under \ap{the noise-only hypothesis, the statistic} is distributed as $\sigma_n^{2}\chi^{2}_{N}$, \sbl{where \mbox{$\sigma_n$} is the per-sample noise standard deviation and} $\chi^{2}_{N}$ \ap{denotes a} chi-square \ap{random variable with} $N$ degrees of \ap{freedom. In the presence of} a Gaussian signal \ap{with} power $\sigma_s^{2}$\ap{, the statistic is distributed as} $(\sigma_n^{2}+\sigma_s^{2})\chi^{2}_{N}$. \ap{Define} $\mathrm{snr}=\sigma_s^{2}/\sigma_n^{2}$ \ap{as} \vf{the in-band power ratio}\ap{, expressed as a} linear quantity rather than \ap{in decibels}, and standardize \ap{the statistic under} the \ap{null hypothesis.}

\ap{For the value used in this work} ($N=500$)\ap{,} we adopt throughout the large-$N$ normal approximation to the standardized \ap{chi-square distribution. Under this approximation,} a threshold $\lambda$ \ap{yields} \sbl{the per-dwell false-alarm probability} $q\sblx{:=\mathbb{P}(\text{statistic}>\lambda\mid\text{noise only})}\simeq Q(\lambda)$, where $\Phi$ \ap{denotes} the standard normal \ap{cumulative distribution} function, $\phi$ its \ap{density function,} and $Q=1-\Phi$ its \ap{upper-tail function. Under} the \ap{alternative hypothesis, the standardized statistic} has mean $d$ and standard deviation $1+\mathrm{snr}$, \ap{with deflection:}
\begin{equation}
d=\frac{N\sigma_s^{2}}{\sigma_n^{2}\sqrt{2N}}=\sqrt{\Delta t\,W}\,\mathrm{snr}.
\label{eq:defl}
\end{equation}
Hence, \sbl{the per-dwell detection probability is} $p\simeq Q\bigl(\frac{\lambda-d}{1+\mathrm{snr}}\bigr)$, which reduces to $Q(\lambda-d)$ for $\mathrm{snr}\ll1$. Following the sonar convention~\cite{Urick1983}\ap{, we} define the detection threshold $\mathrm{DT}$ as the in-band $\mathrm{snr}$, \ap{expressed in} decibels, at which $p=1/2$. \ap{Since \mbox{\mbox{$Q(0)=1/2$}}, the} argument of $Q$ \ap{must vanish. Thus, the denominator} $1+\mathrm{snr}$ \ap{does not affect this condition, yielding} $d=\lambda$ \ap{and hence:}
\begin{equation}
\mathrm{DT}(q)\simeq10\log_{10}\!Q^{-1}(q)-5\log_{10}(\Delta t\,W).
\label{eq:DT}
\end{equation}
Equation~\eqref{eq:DT} \ap{establishes} the mapping between detector sensitivity and false-alarm probability. \ap{Lowering} threshold improves \ap{the detection probability} at every range but \ap{increases} the \ap{false-alarm} rate. At the operating points \ap{considered} below, the normal tail \ap{approximation underestimates} the exact $\chi^{2}_{N}$ \ap{false-alarm probability by a factor of} \mbox{$1.5$}-\mbox{$1.6$}\ap{, corresponding to a} $0.24$-$0.43$\,dB \ap{difference in} $\mathrm{DT}$\ap{. Thus,} the equal-alarm comparison is exact \ap{under the adopted approximation} rather than \ap{under} the exact detector law. For fixed processing parameters\ap{, the} second \ap{term in\mbox{~}\mbox{\eqref{eq:DT}}} is constant\ap{; hence,} the dependence on $q$ is entirely through $10\log_{10}Q^{-1}(q)$. Collecting \ap{all terms} in the passive sonar equation that \ap{are independent of the detection} threshold:
\begin{equation}
\begin{split}
A_m(u;\omega)=\ &\mathrm{SL}-\mathrm{TL}(u\!\to\!c_m;\omega)\\[-2pt]
&-\mathrm{NL}(c_m;\omega)+\mathrm{DI}+5\log_{10}\!(\Delta t W),
\end{split}
\label{eq:Afield}
\end{equation}
the mean signal excess \ap{at} node $m$ \ap{for a target at} position $u$ is $A_m(u;\omega)-10\log_{10}Q^{-1}(q)$. \ap{This quantity} fluctuates \ap{due to variations in} propagation and target \ap{aspect. Following} the standard planning \ap{model, we treat these fluctuations as} Gaussian with standard deviation $\sigma$ and \ap{declare} a detection when \ap{the signal} excess is positive\ap{. The resulting} per-dwell detection probability \ap{is:}
\begin{equation}
p_m(u;q)=\Phi\!\left(\frac{A_m(u;\omega)-10\log_{10}Q^{-1}(q)}{\sigma}\right).
\label{eq:pdet}
\end{equation}

\subsection{Barrier strength}

We \ap{adopt} the standard glimpse model \ap{from} search theory\ap{. A target traversing the region} at speed $v$ presents one conditionally independent detection opportunity per dwell\ap{. Thus,} over an arclength element $\dd s$, node $m$ misses \ap{the target} with probability $(1-p_m)^{\dd s/(v\Delta t)}$. \ap{Although successive} dwells may be correlated in practice\ap{, conditional} independence \ap{provides} a tractable approximation, \ap{and is} applied consistently across all designs \ap{considered. The resulting} \vf{per-arclength detection hazard is:}
\begin{equation}
\haz_m(u;q)=\frac{-\ln\bigl(1-p_m(u;q)\bigr)}{v\,\Delta t}\qquad [\mathrm{m}^{-1}],
\label{eq:hazard}
\end{equation}
and, under conditional independence across nodes and dwells, the accumulated negative-log miss probability along a crossing $\gamma$ is additive:
\begin{equation}
\begin{split}
\psi_m(\gamma;q)&=\int_{\gamma}\haz_m(u;q)\dd s,\\[-2pt]
\Bs(x,q)&=\inf_{\gamma\in\Gpath}\ \sum_{m}x_m\,\psi_m(\gamma;q_m).
\end{split}
\label{eq:psi}
\end{equation}
We \ap{refer to} $\Bs$ \ap{as the} \emph{barrier strength}. The worst-case detection probability is $1-e^{-\Bs}$\ap{. Accordingly, the deployment constitutes} an $\varepsilon$-barrier, \ap{meaning that} every crossing is detected with probability at least $1-\varepsilon$, \ap{if and only if} $\Bs\ge\tau:=\ln(1/\varepsilon)$. In the binary \ap{limit, where} $p_m\to1$ inside a sensing region and $p_m\to0$ \ap{outside it,} the formulation \ap{reduces to} the classical intersection requirement for crossings that \ap{intersect a sensing} region over positive arclength.

\subsection{Design problem}

Let $\bar q$ \ap{denote the maximum} admissible false-alarm \ap{probability per node. For a deployment \mbox{$x$}, define the feasible set of false-alarm allocations as} $\Qset(x)=\{q\in[0,\bar q]^{M}:\sum_m x_m q_m\le B\}$, \ap{where \mbox{\mbox{$B$}} is the network-wide false-alarm budget. The joint design problem is then:}
\begin{equation}
\max_{x\in\{0,1\}^{M},\ \sum_m x_m\le n}\ \ \max_{q\in\Qset(x)}\ \ \sblx{\Bs(x,q)}.
\label{eq:P}
\end{equation}
Fixing $q$ \ap{at} the uniform value $B/n\le\bar q$ and \ap{eliminating} the inner maximization \ap{reduces the problem to} a max-min exposure placement formulation~\cite{Amaldi2012}, which we use as the fixed-operating-point baseline.

\section{Structure of the Allocation Problem}
\label{sec:structure}


\ap{The subsequent analysis hinges} on the curvature of the map $q\mapsto\haz_m(u;q)$, \ap{which characterizes} how the \ap{detection hazard} at a fixed \ap{location} $u$ \ap{varies with the false-alarm} allocation $q$\ap{. This map is not globally concave, but the region in which concavity fails can be characterized exactly. Two competing effects determine its curvature. First, relaxing the detection threshold increases} the signal excess with diminishing returns, \ap{since} $\lambda=Q^{-1}(q)$ \ap{varies increasingly slowly as \mbox{$q$} increases. Second,} the hazard is \emph{convex} \ap{in the} signal excess, \ap{since} $-\ln(1-p)$ \ap{diverges} as $p\to1$. Concavity holds \ap{precisely} when the first effect \ap{dominates} the second\ap{. The following analysis quantifies the tradeoff.} \ap{Define} $h(z)=\phi(z)/Q(z)$ \ap{as the standard normal hazard rate and:}
\begin{equation}
r(z)\ :=\ h(z)-z\ =\ \frac{\phi(z)}{Q(z)}-z\ >\ 0 ,
\label{eq:rdef}
\end{equation}
which is strictly decreasing, with $r(z)\sim|z|$ as $z\to-\infty$ and $r(z)\sim 1/z$ as $z\to+\infty$.

\begin{theorem}
\label{thm:concave}
\ap{Fix a target location} $u$ and a node $m$\ap{. Let} $q<Q(1)=0.159$, and \ap{define} $\lambda=Q^{-1}(q)>1$, $z=\bigl(A_m(u)-10\log_{10}\lambda\bigr)/\sigma$\ap{,} $\kappa=10/(\sigma\ln 10)$. \ap{Then the map} $q\mapsto\haz_m(u;q)$ \ap{is strictly increasing. Moreover,} it is concave at $q$ if \ap{and only if:}
\begin{equation}
\kappa\,r(z)\le\lambda^{2}-1 .
\label{eq:curv}
\end{equation}
\vspace{1pt}
\end{theorem}

\begin{proof}
\ap{Write} $z=A_m(u)/\sigma+G(q)$\ap{,} $G(q)=-\kappa\ln\lambda$, so \ap{that all dependence on \mbox{$q$} is contained} in $G$, and $\haz=-\ln Q(z)/(v\Delta t)$. Differentiating \ap{with respect to} $z$ gives $\partial_z\haz\propto h(z)$\ap{. Using} the identity $h'(z)=h(z)\,(h(z)-z)=h(z)\,r(z)$\ap{, we obtain} $\partial^{2}_{z}\haz\propto h(z)\,r(z)>0$\ap{. Thus,} the hazard is \ap{strictly increasing} and convex in $z$. \ap{Next, differentiating} $G$, \ap{using} $\dd\lambda/\dd q=-1/\phi(\lambda)$ and $\phi'=-\lambda\phi$, \ap{yields:}
\[
G'=\frac{\kappa}{\lambda\phi(\lambda)}>0,
\qquad
G''=\frac{\kappa\,(1-\lambda^{2})}{\lambda^{2}\phi(\lambda)^{2}}<0
\ \text{ for }\lambda>1.
\]
Therefore, $G$ is \ap{strictly increasing} and concave\ap{, and consequently,} $\haz$ \ap{is strictly increasing} in $q$. \ap{Applying the} chain rule gives $\partial^{2}_{q}\haz\propto \apx{h(z)\bigl[r(z)\,(G')^{2}+G''\bigr]}$\ap{. Since} $\apx{h(z)}>0$\ap{, concavity with respect to \mbox{$q$} holds if and only if} $\apx{r(z)\,(G')^{2}}\le-G''$\ap{. Substituting the expressions for \mbox{$G'$} and \mbox{$G''$} and simplifying} yields~\eqref{eq:curv}.
\end{proof}

The hazard rate \ap{cancels from the concavity condition,} which \ap{explains} why~\eqref{eq:curv} \ap{depends only on} \ap{\mbox{$r(z)$}}. Since $r(z)$ decreases \ap{with} $z$ while $\lambda^{2}-1$ decreases \ap{with} $q$, the \ap{condition has a natural interpretation:} concavity holds \ap{when the target signal is sufficiently strong} relative to \ap{the looseness of the detection threshold.} For each fixed $q$\ap{, concavity} fails below \ap{a corresponding} threshold in $A$\ap{. In this regime,} $r(z)\approx|z|$ grows without bound \vf{as \mbox{\mbox{$z\to-\infty$}}, and the hazard becomes highly sensitive} in relative terms, \ap{although its absolute magnitude remains} negligible under the cap \ap{adopted here. Corollary\mbox{~}}\ref{cor:trunc} \ap{exploits precisely this observation.}

\begin{corollary}
\label{cor:trunc}
Fix a cap $\bar q<Q(1)$\ap{, and} choose $\underline A$ such that~\eqref{eq:curv} holds for every $A\ge\underline A$ and every $q\in(0,\bar q]$. Let $\chi:\Rr\to[0,1]$ be continuous, nondecreasing, and independent of $q$, with $\chi(A)=0$ for $A\le\underline A$\ap{, and} $\chi(A)=1$ for $A\ge A_1$, \ap{for some finite \mbox{$A_1>\underline A$}. Define} $\tilde\haz_m(u;q)=\chi\bigl(A_m(u)\bigr)\haz_m(u;q)$. \ap{Then, for every crossing \mbox{$\gamma$}, the map} $q\mapsto\tilde\psi_m(\gamma;q)$ is concave on $(0,\bar q]$\ap{. Since} $\tilde\haz_m(u;q)\to0$ as $q\to0^{+}$, the continuous extension $\tilde\haz_m(u;0):=0$ \ap{extends this concavity to the entire interval} $[0,\bar q]$\ap{. Consequently,} $q\mapsto\tilde\Bs(x,q)$ \ap{is also concave, since it is the} infimum of concave functions. The truncation is conservative, and uniformly over $\Qset(x)$:
\begin{equation}
\begin{gathered}
0\ \le\ \Bs(x,q)-\tilde\Bs(x,q)\ \le\ \delta,\\[-2pt]
\delta:=n\,\ell_{\max}\sup_{A\le A_1}
\bigl[1-\chi(A)\bigr]\haz(A;\bar q).
\end{gathered}
\label{eq:truncbound}
\end{equation}
\vspace{1pt}
\end{corollary}

\begin{proof}
\ap{Where} $\chi=0$, \ap{the truncated hazard vanishes identically} \ap{as a function of \mbox{$q$}. Elsewhere,} it is \ap{a nonnegative,} $q$-independent scaling of a \ap{function that is} concave in $q$\ap{. Concavity is preserved under such scaling, as well as under integration and pointwise} infimum. \ap{For\mbox{\mbox{~}}\mbox{\mbox{\eqref{eq:truncbound}}}, the lower bound follows directly from} $\tilde\haz_m\le\haz_m$. For the \ap{upper bound,} at most $n$ nodes are \ap{active, and the hazard discarded by each active node per unit arclength is} $[1-\chi(A_m)]\haz_m\le\sup_{A\le A_1}[1-\chi(A)]\haz(A;\bar q)$ \ap{since the discarded term vanishes for \mbox{$A\ge A_1$} and} $\haz$ is nondecreasing \ap{in both} $A$ and $q$\ap{. Since every admissible crossing has arclength at most \mbox{$\ell_{\max}$}, the total hazard discarded along a crossing is at most \mbox{$\delta$}. This bound holds for every admissible crossing \mbox{$\gamma$} and every \mbox{$q\in\Qset(x)$}, so taking the corresponding infima yields\mbox{~}\mbox{\eqref{eq:truncbound}}.}
\end{proof}

\ap{Choosing} $\chi$ \ap{to be continuous,} rather than an indicator \ap{function, ensures that} $\tilde\haz$ \ap{remains continuous} on $\overline\D$, \ap{as required in Theorem\mbox{~}}\ref{thm:saddle} for weak-$*$ continuity. On a finite graph, \ap{however, the sharp cutoff} $\chi=\mathbf{1}\{A\ge\underline A\}$ instead \ap{yields} a $q$-independent active set, so the allocation \ap{problem remains} concave \ap{in that setting as well. Thus, the continuous taper} is needed only for Theorem~\ref{thm:saddle}. Section~\ref{sec:case} reports the truncation parameters, the resulting bound $\delta$, and a verified curvature margin, \ap{thereby ensuring that} concavity is \ap{explicitly verified} rather than assumed.

\subsection{The saddle point and the water-filling condition}

A length bound alone does not \ap{ensure that} arbitrary parameterizations of $\Gpath$ \ap{are equi-}Lipschitz\ap{. We therefore} work directly with the object \ap{required for the analysis.} Each \ap{crossing \mbox{$\gamma$}} induces an \emph{occupation measure} $\Theta_\gamma$ on $\overline{\D}$ \ap{defined by} $\apx{\int_{\overline{\D}}} f\dd\Theta_\gamma=\int_\gamma f\dd s$\ap{. Each} such measure is positive \ap{and has} total mass \sbl{equal to the arclength} $\ell(\gamma)\le\ell_{\max}$. Let $\mathcal{M}$ \ap{denote} the weak-$*$ closed convex hull of $\{\Theta_\gamma:\gamma\in\Gpath\}$\ap{. The set \mbox{$\mathcal{M}$}} is convex and, being a weak-$*$ closed set of positive measures of uniformly bounded mass on \ap{the compact domain \mbox{\mbox{$\overline{\D}$}}, is} weak-$*$ compact \ap{by the} Banach-Alaoglu \ap{theorem}\sbl{\mbox{~}\mbox{\cite{Rudin1991}}.} Selecting $\Theta\in\mathcal{M}$ \ap{can be interpreted as} a relaxed randomization \ap{over crossings.}

\ap{Since \mbox{$J$} in}~\eqref{eq:Jdef} is affine and weak-$*$ continuous in $\Theta$, its infimum over $\{\Theta_\gamma\apx{:\gamma\in\Gpath}\}$ equals its infimum over their convex hull\ap{, and by continuity,} over \ap{its} weak-$*$ closure\ap{. Compactness} of $\mathcal{M}$ \ap{then ensures that this last infimum is attained. Hence,} $\inf_{\gamma\in\Gpath}\sum_m x_m\tilde\psi_m(\gamma;q_m)=\min_{\Theta\in\mathcal{M}}J(q,\Theta)$. \ap{Thus, passing} to occupation measures \ap{does not change the optimal value,} although \ap{a} minimizing $\Theta$ need not \ap{correspond to a single crossing.}

\begin{theorem}
\label{thm:saddle}
Fix $x$ and $B>0$\ap{, and assume} the truncated model of Corollary~\ref{cor:trunc}\ap{, with} $\chi$ and every $A_m$ continuous on $\overline\D$\ap{. Define:}
\begin{equation}
J(q,\Theta)=\sum_m x_m\int_{\overline{\D}}\tilde\haz_m(u;q_m)\dd\Theta(u).
\label{eq:Jdef}
\end{equation}
Then, $\max_{q\in\Qset(x)}\min_{\Theta\in\mathcal{M}}J=\min_{\Theta\in\mathcal{M}}\max_{q\in\Qset(x)}J$\ap{, and} a saddle point $(q^{\star},\Theta^{\star})$ exists. For $q\in(0,\bar q]$\ap{, let} $\lambda(q)=Q^{-1}(q)$\ap{,} \ap{define:}
\begin{equation}
D_m(q):=\frac{1}{\lambda(q)\phi(\lambda(q))}\int_{\overline{\D}}\chi\bigl(A_m(u)\bigr)\,h\bigl(z_m(u;q)\bigr)\dd\Theta^{\star}(u)\apx{.}
\label{eq:margval}
\end{equation}
\ap{Extend this definition} to $q=0$ \ap{by setting} $\tilde\haz_m(u;0):=0$ \ap{and defining} the right derivative $D_m(0^+):=\lim_{q\to0^{+}}D_m(q)\in[0,+\infty]$. Then there \ap{exists} $\nu\ge0$ such that, for every selected node ($x_m=1$), \ap{defining} $D_m\apx{=} D_m(q^\star_m)$ \ap{when \mbox{$q^\star_m>0$} and \mbox{$D_m=$}} $D_m(0^+)$ when $q^\star_m=0$:
\begin{equation}
D_m\ \begin{cases}
\le\nu, & q^{\star}_m=0,\\
=\nu, & 0<q^{\star}_m<\bar q,\\
\ge\nu, & q^{\star}_m=\bar q,
\end{cases}
\label{eq:waterfill}
\end{equation}
and $\nu>0$ only if the \ap{false-alarm budget constraint is active.}
\end{theorem}

\begin{proof}
$J$ is concave in $q$ by Corollary~\ref{cor:trunc} \ap{and affine} \ap{in \mbox{$\Theta$}. Moreover,} since $\tilde\haz_m$ is continuous on $\overline\D$\ap{, \mbox{$J$} is} weak-$*$ continuous in $\Theta$. The set $\Qset(x)$ is convex and compact, and $\mathcal{M}$ is convex and weak-$*$ compact\ap{. Therefore,} Sion's minimax theorem~\cite{Sion1958} \ap{yields the minimax equality and the existence of} a saddle point.

At a saddle \ap{point,} $q^{\star}$ maximizes the concave map $q\mapsto J(q,\Theta^{\star})$ over $\Qset(x)$\ap{. This map} is differentiable \ap{in} the interior and admits one-sided derivatives at the endpoints\ap{. Hence,}~\eqref{eq:waterfill} is the KKT system\ap{, written in} Fenchel \ap{form, associated with} the single coupling constraint $\sum_m x_mq_m\le B$ with multiplier $\nu$\ap{. Using the chain-rule expression from} Theorem~\ref{thm:concave}:
\[
\frac{\partial\haz_m}{\partial q}
=\frac{1}{v\Delta t}\cdot h(z_m)\cdot\frac{\kappa}{\lambda_m\phi(\lambda_m)},
\]
and absorbing the common factor $\kappa/(v\Delta t)$ into $\nu$\ap{, yields\mbox{~}\mbox{\eqref{eq:waterfill}}.}
\end{proof}

\begin{remark}
Condition~\eqref{eq:waterfill} has the form of a water-filling rule. The factor $1/(\lambda\phi(\lambda))$ \ap{converts a small increase in false-alarm allocation into an increase in} signal excess, \ap{while} the inverse Mills ratio \ap{\mbox{$h(z)=\phi(z)/Q(z)$}} \ap{converts that increase} into marginal \ap{detection hazard. Their} product, integrated \ap{with respect to} $\Theta^{\star}$, \ap{gives} the marginal barrier gain per unit \ap{increase in false-alarm allocation, which is} equalized across selected nodes at an interior optimum. The \ap{boundary case \mbox{$q=0$} is sharp. As} $q\to0^{+}$\ap{, the factor \mbox{$1/(\lambda\phi(\lambda))$}} diverges faster than $h(z_m)$ vanishes\ap{. Thus,} $D_m(0^+)=+\infty$ unless $\Theta^{\star}$ \ap{assigns zero} mass to $\{\apx{u:}\chi(A_m\apx{(u))}>0\}$\ap{. Therefore,} a zero allocation is optimal only for such a node, \ap{in which case} $D_m\equiv0$.
\end{remark}

\section{Certified Joint Design}
\label{sec:algorithm}

\subsection{A perspective reformulation}

Problem~\eqref{eq:P} couples \ap{the} binary placement variable $x$ \ap{and the continuous allocation variable} $q$ through the products $x_mq_m$. \ap{Introducing the} change of variable $b_m=x_mq_m$, \ap{where \mbox{$b_m$} denotes the false-alarm allocation assigned to site \mbox{$m$}, removes this bilinear coupling.}

\begin{proposition}
\label{prop:persp}
\ap{Consider} the truncated model of Corollary~\ref{cor:trunc}\ap{, and let} $\tilde\Bs^{\star}$ \ap{denote the optimal} value of~\eqref{eq:P} \ap{when} $\Bs$ \ap{is replaced} by $\tilde\Bs$. For each $\gamma$ and $m$, the map $(x_m,b_m)\mapsto x_m\tilde\psi_m(\gamma;b_m/x_m)$, extended \ap{by continuity to} $0$ at \ap{\mbox{$(x_m,b_m)=(0,0)$},} is the perspective of the \ap{concave function} $\tilde\psi_m(\gamma;\cdot)$\ap{. It is therefore} jointly concave on $\{x_m\in[0,1],\ 0\le b_m\le\bar q\,x_m\}$. For every tangent point $t\in(0,\bar q]$:
\begin{equation}
x_m\tilde\psi_m(\gamma;b_m/x_m)\ \le\ \tilde\psi_m(\gamma;t)\,x_m+\tilde\psi_m'(\gamma;t)\bigl(b_m-t\,x_m\bigr).
\label{eq:tangent}
\end{equation}
\end{proposition}

\begin{proof}
Joint concavity of the perspective of a concave function is standard~\cite[Sec.~3.2.6]{Boyd2004}\ap{. At a point} $(x_0,b_0)$ \ap{with \mbox{$x_0>0$} and \mbox{$t=b_0/x_0$}, a supergradient} has components $\partial_b=\tilde\psi'(t)$ and $\partial_x=\tilde\psi(t)-t\tilde\psi'(t)$\ap{. The corresponding supergradient inequality rearranges directly} to~\eqref{eq:tangent}.
\end{proof}

\begin{remark}
\label{rem:milp}
\vf{Summing~\mbox{\eqref{eq:tangent}} over $m$ gives one linear cut for each crossing and tangent point. Maximizing $s$ over $x\in\{0,1\}^{M}$, $b$ and $s$, subject to these cuts and the constraints $\sum_mx_m\le n$, $\sum_m b_m\le B$ and $b_m\le\bar q\,x_m$, gives a mixed-integer linear program. Restricting the infimum in~\mbox{\eqref{eq:psi}} to a finite set of crossings can only increase its value, and replacing each concave function $\tilde\psi_m$ by finitely many tangent upper bounds can only increase it further. Hence, the resulting MILP value $\overline{\Bs}$ satisfies $\tilde\Bs^{\star}\le\overline{\Bs}$. Conversely, any feasible design $(x,q)$, when evaluated using the exact truncated barrier strength, yields a value $\underline{\Bs}$ satisfying $\underline{\Bs}\le\tilde\Bs^{\star}$. Finally,~\mbox{\eqref{eq:truncbound}} gives the pointwise bound $\tilde\Bs\le\Bs\le\tilde\Bs+\delta$, and therefore $\tilde\Bs^{\star}\le\Bs^{\star}\le\tilde\Bs^{\star}+\delta$. Combining these inequalities yields $\underline{\Bs}\le\Bs^{\star}\le\overline{\Bs}+\delta$.}
\end{remark}

This \ap{leads to} a cutting-plane scheme in the sense of Kelley~\cite{Kelley1960}. \ap{Starting from} a deployment \ap{with} the uniform operating point $q=B/n$, each \ap{iteration} solves the concave allocation \ap{problem of} Theorem~\ref{thm:saddle} for the \ap{current deployment} $x$, \ap{invokes} a shortest-crossing oracle \ap{to evaluate the resulting design and generate} a new crossing, \ap{and then} re-solves the master MILP \ap{using} all crossings and tangent points generated \ap{thus} far. \ap{A tangent point is added} at each positive \ap{allocation encountered.}

The oracle \ap{used here} is graph-based\ap{; consequently, the resulting bounds} are graph-discretized quantities, \ap{whose relation to the continuous-domain problem is certified} separately by Proposition~\ref{prop:cert}. \ap{If} the master \ap{MILP is terminated} early, its dual bound\ap{, rather than the objective value of the incumbent solution, is used to preserve the validity of the upper bound.}

\subsection{Certifying the continuum, not the grid}

Restricting the infimum in~\eqref{eq:psi} to grid polylines can only \ap{increase its value. Consequently,} the truncated graph value $\Bs_G$ \ap{provides, up to the truncation error} $\delta$\ap{, an upper bound on the continuum value \mbox{$\Bs$} and therefore cannot}\ap{, by itself, establish the feasibility condition} $\Bs\ge\tau$\ap{. Feasibility must instead be established through a separate} certificate.

\ap{For a fixed design, define the aggregate truncated hazard field as} $\haz(u):=\sum_m x_m\,\tilde\haz_m(u;q_m)$\ap{.} The following \ap{result provides a bound directly on} the continuum \ap{field,} independently of the \ap{discretization or the solver used to obtain the candidate design.}

\begin{proposition}
\label{prop:cert}
Let $\haz$ be continuous on $\overline\D$\ap{, and} let $V$ be continuous and piecewise affine on a conforming simplicial partition of $\D$\ap{. Suppose that} $V\le0$ on $\partial_{\mathrm{in}}$ and $\|\nabla V|_S\|\le\inf_{S}\haz$ on every simplex $S$. Then $\Bs\ge\min_{\partial_{\mathrm{out}}}V$. \ap{Conversely, the exact cost of any admissible crossing provides an upper bound on} $\Bs$.
\end{proposition}

\begin{proof}
Take an admissible $\gamma$ parameterized by arclength. On each portion of $\gamma$ inside a simplex $S$, $|\tfrac{\dd}{\dd s}V(\gamma(s))|\le\|\nabla V|_S\|\le\haz(\gamma(s))$. Since $V$ is continuous across \ap{simplex interfaces,} these \ap{inequalities can be integrated} piecewise along the \ap{entire crossing. Since \mbox{$\gamma(0)\in\partial_{\mathrm{in}}$}, where \mbox{$V\le0$}, and \mbox{$\gamma(1)\in\partial_{\mathrm{out}}$}, this yields:}
\begin{equation}
\begin{gathered}
V(\gamma(1))-V(\gamma(0))\ \le\ \int_\gamma\haz\dd s,\\[-2pt]
\min_{\partial_{\mathrm{out}}}V\ \le\ V(\gamma(1))\ \le\ \int_\gamma\haz\dd s .
\end{gathered}
\label{eq:certchain}
\end{equation}
\ap{Taking the infimum over \mbox{$\gamma\in\Gpath$} and using \mbox{$\tilde\haz_m\le\haz_m$} gives \mbox{$\Bs\ge\min_{\partial_{\mathrm{out}}}V$}. The} upper bound follows by evaluating \ap{the exact cost of any admissible} crossing.
\end{proof}

Stating the condition \ap{on each} simplex, rather than almost everywhere, \ap{makes the certificate} finitely checkable, \ap{and also covers crossings that run along simplex} interfaces. \ap{Here,} $\haz$ is the multilinear interpolant of its nodal \ap{values, while} $V$ is piecewise \ap{affine} on a Freudenthal subdivision, \ap{so that} $\nabla V$ is constant on \ap{each simplex} $S$. \ap{Within the cell containing \mbox{$S$}, the interpolant \mbox{$\haz$}} is a convex combination of the eight nodal \ap{values and is therefore bounded below by their minimum. Hence,} it suffices to verify $\|\nabla V|_S\|$ against \ap{the minimum of these eight corner values.}

\ap{This} test is conservative \ap{when} $\haz$ varies sharply within a cell, \ap{and} the refinement study \ap{in} Section~\ref{sec:certresults} \ap{quantifies this conservatism.} The candidate $V$ \ap{is obtained} from an anisotropic fast-marching field capped at its minimum exit \ap{value. This capping eliminates} superlevel variation before the simplexwise \ap{verification while leaving} $\min_{\partial_{\mathrm{out}}}V$ unchanged. Scaling by the largest verified ratio then yields an admissible certificate \ap{over the entire domain, with no region excluded from the verification.} For the upper bound, the \ap{multilinear interpolant restricted to each graph segment is a cubic polynomial; hence, two-point Gauss-Legendre quadrature evaluates its line integral exactly.}

\section{Case Study}
\label{sec:case}

\subsection{Setup}

The environment is a deep-water chokepoint \CHANW{}\,km wide and \CHANDEP{}\,m deep, with a \vf{crossing band of \mbox{\CHANL{}\,km}.} Sound-speed profiles \ap{are based on} the canonical Munk profile~\cite{Jensen2011}\ap{, using} deep-channel (``winter'') and warm-surface-layer (``summer'') forms\ap{. Transmission} loss at \FREQ{}\,Hz \ap{is computed using} incoherent ray-density summation with Rayleigh bottom loss and Thorp absorption \ap{and validated to within} \ISOERR{}\,dB against isovelocity spherical spreading.

\ap{A} shipping lane \ap{centered} at $\xi=\LANEX{}$\,km\ap{, with a} Gaussian \ap{scale of} \LANEW{}\,km, \ap{increases the ambient noise} level by up to \LANEDB{}\,dB, so otherwise identical sites \ap{can have different detection performance.} Candidate sites are \NSLOTS{} lateral slots at \NDEPTHS{} depths, giving $M=\NCAND{}$ \ap{candidate sites} and $\binom{\NCAND{}}{\NBUD{}}=\NSUBSETS{}$ \ap{possible deployments} for a budget of $n=\NBUD{}$ nodes. The \ap{network-wide alarm} budget is \ALARMRATE{} false alarms per hour\ap{, corresponding to} $B=\BUDGET{}$\ap{. The uniform operating point is therefore} $q=\QUNIF{}$\ap{, corresponding to} $\mathrm{DT}=\DTUNIF{}$\,dB. The \ap{detection requirement} is $\varepsilon=\EPS{}$, \ap{giving} $\tau\apx{=\ln(1/\varepsilon)}=\TAU{}$. The detector operates with $\Delta t=\DWELL{}$\,s and \ap{a bandwidth satisfying} \vf{\mbox{$\Delta t\,W=\TW{}$}}\ap{, with directivity} index $\mathrm{DI}=\DI{}$\,dB and signal-excess standard deviation $\sigma=\SIGMA{}$\,dB. The intruder speed \ap{is fixed at} $v=\SPEED{}$\,m/s \ap{as a scenario parameter, rather than optimized over a range of possible speeds to determine the worst case.}

Truncation uses $\bar q=\QMAX{}$\ap{, corresponding to} a per-node cap of \QMAXFA{} false alarms per \ap{hour, and} $\underline A=\ACUT{}$\,dB, \ap{at which the} single-dwell \ap{detection probability} is below $10^{-9}$\ap{. The minimum} curvature margin over \ap{the entire} field, all four ocean \ap{states,} and all $q\in(0,\bar q]$ is $+0.07$. With $\ell_{\max}=240$\,km, twice the channel width, a $3$\ap{-}dB taper \ap{yields} $\delta\le3.7\times10^{-3}$, \ap{whereas} the sharp \ap{cutoff} used here gives $\delta=\TRUNCBOUND{}$.

The grid \ap{contains} $\GRIDDESC{}$ nodes, with \ap{a depth resolution of} \HZ{}\,m because the field varies \ap{more rapidly} \ap{with} depth than laterally. Refining the \ap{depth-cell size} from \RESHZ{}\,m to \RESHZFINE{}\,m \ap{changes} the exact crossing \ap{cost by \mbox{\RESSPAN{}\,\%} in total}, relative to the finest grid, \ap{with only a} \RESLAST{}\,\% \ap{change over} the final \ap{refinement from \mbox{\HZ{}\,m} to \mbox{\RESHZFINE{}\,m}. This latter change is} an order of magnitude \ap{smaller than the effect of alarm allocation.}

Two modeling artifacts \ap{identified} by the path optimizer were corrected \ap{before the final evaluation. First,} tent-kernel depth binning collects \ap{only half} a cell at the boundaries\ap{;} dividing by a full cell height \ap{therefore introduced an artificial 3-dB shadow at these boundaries. Second,} range-parameterized ray integration over $\pm80^\circ$ carried no energy along near-vertical paths\ap{. To correct this artifact,} transmission loss is capped \ap{at the} direct-path \ap{spreading loss} within $2$\,km. The worst-case graph crossing in every design and ocean \ap{state considered} is a straight traverse of \CHANL{}\,km \ap{along either} the seabed or the surface\ap{. Consequently,} the endurance bound $\ell_{\max}\apx{=240\,\mathrm{km}}$ is slack by a factor of $60$ and \ap{is never} active.

\subsection{Results}

\begin{table}[t]
\caption{Designs \ap{with} $n=6$ sensors \ap{and a network-wide budget of} 6 false alarms per hour. $\Psi_G$ \ap{and} $P_{{\rm det},G}=1-e^{-\Psi_G}$ \ap{are reported for the} nominal ocean state.}
\label{tab:designs}
\centering\small
\setlength{\tabcolsep}{4pt}
\begin{tabular}{lccc}
\toprule
design & $\Psi_G$ & $P_{{\rm det},G}$ & $\operatorname{LCVaR}_{0.4}$ \\
\midrule
evenly spaced, uniform $q$ & 0.87 & 0.583 & 0.14 \\
max--min placement, uniform $q$ & 3.04 & 0.952 & 0.54 \\
\textbf{same placement, optimized $q_m$} & 3.49 & 0.970 & 0.64 \\
\bottomrule
\end{tabular}
\end{table}

\ap{Placement provides the largest improvement in this case study. As shown in} Table~\ref{tab:designs}, moving from \ap{the} evenly spaced layout to the max-min exposure optimum \ap{increases} $\Bs_G$ from \PSISPREAD{} to \PSIUNIF{}. Since $\Bs_G$ \ap{overestimates the continuum value} $\Bs$, the evenly spaced layout \ap{can be ruled out:} $\Bs\le\Bs_G+\delta\apx{\simeq}\PSISPREAD{}<\tau$, so it \ap{cannot constitute an} $\varepsilon$-barrier under any refinement. The converse does not follow for the optimized \ap{layout, since an upper bound exceeding \mbox{$\tau$} does not establish feasibility; the continuum certificate in} Section~\ref{sec:certresults} \ap{provides the required lower bound.}

\ap{This approximately} \PLACEFACTOR{}-fold \ap{improvement} is not a contribution \ap{of the present work; it results} from the max-min exposure \ap{placement method} of~\cite{Amaldi2012}. The integer \ap{cutting-plane algorithm} closes the graph-discretized placement problem to zero \ap{optimality gap} after \UNIFITERS{} iterations, \ap{thereby certifying} the optimum over all $\NSUBSETS{}$ \ap{possible} deployments.

\ap{Operating-point reallocation improves the graph objective under the same resource budget.} Redistributing the same total \ap{false-alarm budget} across the same six sites \ap{increases} $\Bs_G$ from \PSIUNIF{} to \PSIALLOC{}, a gain of \ALLOCGAIN{}\,\%, \ap{while reducing} the graph-derived worst-path miss surrogate from \MISSUNIF{}\,\% to \MISSALLOC{}\,\%. The allocation subproblem \ap{is solved} to a relative \ap{optimality gap} below $10^{-5}$. Fig.~\ref{fig:main}(a) \ap{illustrates} the concavity \ap{underlying this optimization, while Fig.\mbox{~}\mbox{\ref{fig:main}}(b) shows the resulting allocation.}

\ap{The lower-tail robustness measure exhibits an even larger relative improvement.} The \sbl{lower conditional value at risk at level \mbox{$\alpha$},} \sbl{\mbox{$\operatorname{LCVaR}_{\alpha}:=$}} $\max_t[t-\alpha^{-1}\Ex(t-\Bs_G)_+]$~\cite{RockafellarUryasev2000}, \ap{evaluated at} $\alpha=0.4$ over the four ocean states\ap{ with} probabilities $(0.3,0.3,0.2,0.2)$, rises from \CVARUNIF{} to \CVARALLOC{}, \ap{a gain of} \CVARGAIN{}\,\%. \ap{This quantity} is evaluated \ap{rather than} optimized, so \ap{its improvement is not guaranteed by the allocation procedure; instead, it arises because reallocation provides the greatest benefit in the ocean states} where the barrier is weakest.

The joint placement-and-allocation scheme did not improve on the allocation obtained for the uniform-threshold placement, and its graph-discretized bracket remained wide at $[\JOINTLO{},\ \JOINTHI{}]$. \ap{This experiment cannot determine whether} the two decisions are nearly separate \ap{in this case or whether the master problem simply requires more iterations.}

\begin{figure*}[t]
\centering
\includegraphics[width=0.80\textwidth]{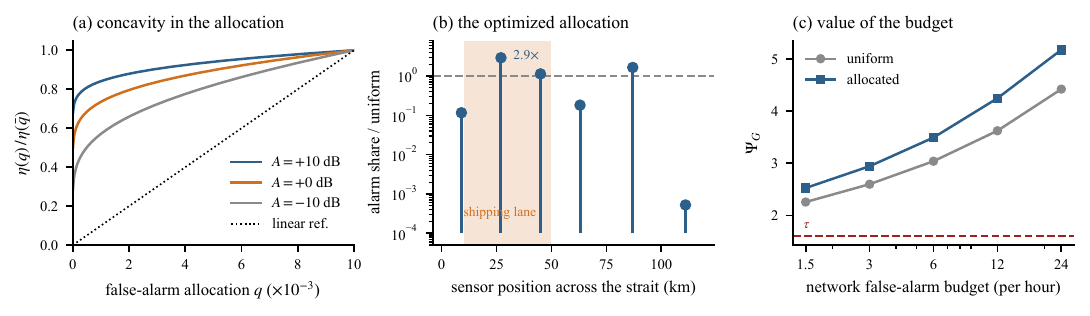}
\caption{(a) The hazard is concave in a node's false-alarm allocation over the retained operating region (Theorem~\ref{thm:concave}). (b) Optimizing the shared budget \ap{yields} a nonuniform allocation across the deployed sensors; the shaded region \ap{indicates} the shipping lane. (c) Sensitivity of the graph barrier strength to the total network false-alarm budget at the fixed max--min placement.}
\label{fig:main}
\end{figure*}

\subsection{Certification}
\label{sec:certresults}

For the multilinearly interpolated field, the uniform and allocated designs have certified brackets $[1.530,\ 3.040]$ and $[1.651,\ 3.494]$, respectively. Each upper endpoint is the exact interpolated-field cost of an explicit crossing. \ap{To obtain} the lower endpoint, an anisotropic fast-marching potential is capped at its minimum exit value and then scaled by $\rho$, the largest ratio \ap{of} a simplex gradient norm \ap{to} the corresponding local hazard lower bound\ap{, with \mbox{$\rho=1.99$} and \mbox{$2.12$} for the uniform and allocated designs, respectively.} After scaling, the certificate inequality holds on every simplex, with no region \ap{excluded. Consequently, the certified lower} bound does not depend on the accuracy of the fast-marching field itself.

\ap{The brackets} above \ap{correspond to} the nominal \ap{25-}m depth mesh, \ap{for which} \ap{the corresponding} lower bound is 47\,\% of the same-mesh \ap{upper bound. Coarsening the mesh} to 100\,m \ap{reduces this} ratio to 16\,\%, \ap{whereas refining it} to 12.5\,m raises \ap{the ratio} to 55\,\%. At the nominal resolution\ap{, the allocated-design} lower bound exceeds $\tau$\ap{, while} the uniform-design bracket straddles it\ap{. Thus, the barrier} requirement is certified for the allocated design only. \ap{The two} brackets overlap, so \ap{they do not establish an ordering between the corresponding} continuum values.

Holding the placement fixed at the uniform-threshold optimum and varying the network false-alarm rate gives Fig.~\ref{fig:main}(c). The logarithmic sweep stops at 24 \ap{false alarms per hour because the} next doubling would exceed the uniform-design cap of 43.2 per hour. \ap{Viewed conversely, with} uniform thresholds, matching the barrier \ap{achieved by reallocation} at 6 false alarms per hour would \ap{require approximately 10 false alarms} per hour, or 1.7 times the operator load. The gain also rises with \ap{the false-alarm} budget, from 11.9\,\% at the tightest rate \ap{considered} to 17.0\,\% at the loosest, where it plateaus as nodes begin to reach the per-node cap.

\section{\sblx{Discussion}}
\label{sec:discussion}

Four limitations bound the \ap{claims of this work. First, detections} are treated as conditionally independent across dwells and nodes; under correlated clutter\ap{, however, the budget itself remains a simple sum by} linearity of expectation. \ap{Second, fusion is not modeled,} and $m$-of-$n$ logic \ap{could allow} individual nodes \ap{to operate at considerably looser thresholds. Third, the continuum} certificate is valid but conservative, and the \ap{joint optimization} bracket \ap{remains wide. Finally, the} numerical results \ap{represent a single case study:} the structural results do not depend on \ap{this particular instance, but the magnitude of the observed improvement} does.

\section{Conclusion}
\label{sec:conclusion}

Detector operating points \ap{constitute} a network resource that can be co-designed with sensor placement. Under a shared false-alarm budget, the problem \ap{admits} a concave inner allocation \ap{over} the retained \ap{operating region,} whose optimum equalizes \ap{the marginal} barrier gain per unit alarm against a worst-case occupation \ap{measure. The joint problem further admits} a perspective structure that \ap{provides certified bounds on its optimum value.} \sbl{On the deep-water case study, optimizing the placement raises the graph barrier strength by a factor of \mbox{\PLACEFACTOR{}} over an evenly spaced layout, and reallocating the same false-alarm budget over that placement raises it by a further \mbox{\ALLOCGAIN{}\,\%}, with no additional sensors and no increase in operator alarm load. The formulation also contains the familiar models as special cases: the requirement reduces to the \mbox{$\varepsilon$}-barrier condition through \mbox{$\Bs\ge\tau$}, and the binary sensing model is recovered in the limit of perfect detection inside a sensing region and none outside.}

Natural extensions \ap{include} $m$-of-$n$ fusion \ap{under} spatially correlated clutter, sequential detection with a budget on time to alarm\ap{, and post-deployment} reallocation as the ocean state is \ap{estimated.}

\bibliographystyle{IEEEtran}
\bibliography{barrier_coverage_refs}

\end{document}